\documentclass[a4paper,UKenglish,cleveref, autoref, thm-restate]{lipics-v2021}

\usepackage{graphicx} 
\usepackage{amssymb}
\usepackage{amsthm}
\usepackage{mathtools}
\usepackage{algorithmic}
\usepackage{booktabs}
\usepackage[table]{xcolor}

\usepackage{cite}

\usepackage{apxproof}
\newtheoremrep{thm}[theorem]{Theorem}
\newtheoremrep{prop}[proposition]{Proposition}

\usepackage{todonotes}

\def\david#1{{\color{magenta}[DW] #1}}



\Copyright{Youri Kaminsky, Benny Kimelfeld, Ester Livshits, Felix Naumann, and David Wajc}

\title{Repairing Databases over Metric Spaces
 with Coincidence Constraints}
 \titlerunning{Repairing Databases over Metric Spaces
with Coincidence Constraints}
\authorrunning{Y. Kaminsky, B. Kimelfeld, E. Livshits, F. Naumann, and D. Wajc}

\author{Youri Kaminsky}{Hasso Plattner Institute, University of Potsdam, Germany}{}{}{}
\author{Benny Kimelfeld}{Technion, Haifa, Israel}{}{}{}
\author{Ester Livshits}{Technion, Haifa, Israel}{}{}{}
\author{Felix Naumann}{Hasso Plattner Institute, University of Potsdam, Germany}{}{}{}
\author{David Wajc}{Technion, Haifa, Israel}{}{}{}
\date{}

\def\e#1{\emph{#1}}

\def\cells{\mathsf{Cells}}
\def\vals{\mathsf{Vals}}
\def\atts{\mathsf{Atts}}

\newcommand{\defeq}{\vcentcolon=}

\def\set#1{\mathord{\{#1\}}}

\def\reals{\mathbb{R}}

\newcommand{\eat}[1]{}

\def\costsum{\kappa_{\mathsf{sum}}}

\def\naturals{\mathbb{N}}

\def\signature{\mathbf{A}}

\def\key{\mathsf{key}}
\def\fk{\sqsubseteq_\mathsf{k}}

\newcolumntype{R}{>{\columncolor[HTML]{D0D3D4}}l}

\newcommand{\rel}[1]{{\normalfont\textsc{#1}}}
\newcommand{\att}[1]{{\normalfont\textsf{#1}}}
\newcommand{\val}[1]{{\normalfont\texttt{#1}}}

\def\qedexample{\hfill$\Diamond$}

\usepackage{framed}
\usepackage[framemethod=tikz]{mdframed}
\usepackage[skins]{tcolorbox}
\newenvironment{wrapper}[1]
{
	\smallskip
	\begin{center}
		\begin{minipage}{\linewidth}
			\begin{mdframed}[hidealllines=true, backgroundcolor=gray!20, leftmargin=0cm,innerleftmargin=0.3cm,innerrightmargin=0.3cm,innertopmargin=0.375cm,innerbottommargin=0.375cm,roundcorner=10pt]
				#1}
			{\end{mdframed}
		\end{minipage}
	\end{center}
	\smallskip
}

\begin{document}

\maketitle

\begin{abstract}

Datasets often contain values that naturally reside in a metric space: numbers, strings, geographical locations, machine-learned embeddings in a Euclidean space, and so on. 
We study the computational complexity of repairing inconsistent databases that violate integrity constraints, where the database values belong to an underlying metric space. The goal is to update the database values to retain consistency while minimizing the total distance between the original values and the repaired ones. 
We consider what we refer to as \emph{coincidence constraints}, which include key constraints, inclusion, foreign keys, and generally any restriction on the relationship between the numbers of cells of different labels (attributes) coinciding in a single value, for a fixed attribute set. 

\smallskip

We begin by showing that the problem is APX-hard for general metric spaces.
We then present an algorithm solving the problem optimally for tree metrics, which generalize both the line metric (i.e., where repaired values are numbers) and the discrete metric (i.e., where we simply count the number of changed values). 
Combining our algorithm for tree metrics and a classic result on probabilistic tree embeddings, we design a (high probability) logarithmic-ratio approximation  for general metrics. We also study the variant of the problem where each individual value's allowed change is limited. In this variant, it is already NP-complete to decide the existence of any legal repair for a general metric, and we present a polynomial-time repairing algorithm for the case of a line metric.
\end{abstract}

\def\opt{\mathit{Opt}}
\def\true{\mathrm{true}}
\def\false{\mathrm{false}}

\def\subseteqs{\mathrel{\subseteq_{\mathsf{s}}}}
\def\subseteqp{\mathrel{\subseteq_{\mathsf{p}}}}



\section{Introduction}

A pervasive problem encountered in databases is \emph{inconsistency}, which means that the database violates some integrity constraints that are expected to hold in reality. Such an anomaly can arise due to mistaken data sources (e.g., Web resources), erroneous data recording (e.g., manual form filling), noisy data generation (e.g., machine learning), imprecise data integration (e.g., faulty entity resolution), and so on. Consequently, a well-studied computational problem that arises with inconsistent data is database \emph{repairing}---suggesting a \emph{minimal intervention} needed to correct the database so that all integrity constraints are satisfied. 
Studies of this general problem differ in the considered types of 
(i)~integrity constraints, 
(ii)~intervention models, 
and 
(iii)~intervention cost measure.

Common types of integrity constraints include functional dependencies~\cite{DBLP:journals/tods/LivshitsKR20,DBLP:conf/icdt/KolahiL09,DBLP:journals/vldb/MiaoZLWC23,DBLP:conf/icdt/GiladIK23,DBLP:conf/sigmod/BohannonFFR05} (including key constraints), the more general denial constraints~\cite{DBLP:journals/is/BertossiBFL08,DBLP:journals/iandc/ChomickiM05,DBLP:conf/icde/ChuIP13}, and inclusion dependencies~\cite{DBLP:conf/sigmod/BohannonFFR05,DBLP:conf/foiks/MahmoodVBN24,DBLP:journals/iandc/ChomickiM05}.
The intervention model refers to the operations applied to repair the database, and these are typically tuple deletion, tuple insertion, and value updates~\cite{DBLP:series/synthesis/2011Bertossi,
DBLP:series/synthesis/2012Fan}.
In this work, we focus on the third and seek a so called \emph{update repair}. For this intervention model, 
finding an optimal repair
is almost always shown to be computationally hard and algorithms developed are therefore mostly heuristic or focus on highly specialized cases~\cite{DBLP:journals/tods/LivshitsKR20,DBLP:conf/icdt/KolahiL09,DBLP:conf/icdt/GiladIK23,DBLP:conf/icdt/SaIKRR19,DBLP:conf/icde/ChuIP13,DBLP:conf/sigmod/BohannonFFR05}.
This state of affairs holds for various intervention cost measures studied, including
the number of changed cells~\cite{DBLP:journals/tods/LivshitsKR20,DBLP:conf/icdt/KolahiL09}, the sum of changes of numerical values~\cite{DBLP:journals/is/BertossiBFL08},
or the sum of user-provided costs between the original and updated values~\cite{DBLP:conf/sigmod/BohannonFFR05,DBLP:conf/icde/ChuIP13}, often cast as probabilities of (independent) changes~\cite{DBLP:conf/icdt/GiladIK23,DBLP:conf/icdt/SaIKRR19}.

To obtain provable guarantees for repair costs in a wide variety of update cost measures, in this work we exploit 
the fact that database values commonly belong to a structured domain, namely a \emph{metric space}, where update costs of cells abide by the laws of a metric space (positivity, symmetry, and the triangle inequality).
The obvious example of such a metric space is the space of numerical values \cite{DBLP:journals/is/BertossiBFL08}. Another simple example is the \emph{discrete metric}, where any two different values are a unit distance apart (hence, the cost of a repair is the number of changed cells~\cite{DBLP:journals/tods/LivshitsKR20,DBLP:conf/icdt/KolahiL09}). 
A more sophisticated example is the distance or travel time between map locations.  Yet another example is textual values and distance given by edit distance, or the distance between textual embeddings in Euclidean space~\cite{DBLP:journals/access/PatilBGN23}.
In fact, several studies have proposed ways of embedding general database cells (including opaque keys) in Euclidean spaces, with the objective that semantically close cells should be mapped to geometrically close vectors, and vice versa~\cite{DBLP:conf/icde/TonshoffFGK23,DBLP:conf/sigmod/CappuzzoPT20,DBLP:conf/sigmod/BordawekarS17}. 
Hence, there is a wealth of value types corresponding to metrics, so a unifying approach for repairing databases with metric values would have wide applications.
%
%
Motivated by this observation, we introduce the following general computational problem.

\begin{wrapper}
\textbf{\underline{Metric Database Repair Problem} (informal; see \Cref{sec:framework})}\\
A \e{metric database} consists of a set of \e{cells}, each with a label (attribute) and a value from a metric space. 
A \e{coincidence constraint} lists, for each value, the allowed combinations of numbers of coincident cells of different labels.
A \emph{repair} of an inconsistent database moves cells between points in the metric space (i.e., an update of cells' values).
The goal is to compute a repair (if one exists) minimizing the total distance moved over all cells. The attribute set is fixed, but all else is given as input, including the (finite) metric~space.
\end{wrapper}

Coincidence constraints generalize (unary) inclusion constraints, key constraints, foreign-key constraints, as well as other constraints on cardinalities (e.g., the number of people associated with a company is at most the number of employees of the company according to some external trusted registry). Moreover, this class of constraints is closed under conjunction, disjunction, and negation. For example, we can phrase a constraint stating that every location of a team includes one driver (Driver.location cell), and either engineers (Eng.location cells) or salespeople (SP.location cells), but not both.  

\subparagraph*{Contributions.}

Our first contribution is a formalization of the problem introduced informally above, together with illustrative examples (\Cref{sec:framework}).
We then begin the complexity investigation of this problem by showing that, in the generality presented here, it is NP-hard and even APX-hard (\Cref{sec:hardness}).

Our main technical contribution is a polynomial-time algorithm for finding an optimal repair when the metric is a \emph{tree metric} (\Cref{sec:tree}). This algorithm implies the tractability of the problem for well-studied special cases of the tree metric, namely the \emph{line metric} (i.e., the usual metric on numeric values) and the \e{discrete metric} where, as mentioned above, the cost of a repair is the number of changed cells. Moreover, combining our algorithm for tree metrics with classic results on randomly embedding a general metric into a tree metric~\cite{DBLP:journals/jcss/FakcharoenpholRT04,DBLP:conf/focs/Bartal96}, we establish a (high-probability) logarithmic-factor approximation for general metrics  (\Cref{sec:tree-to-general}).

Next, we investigate two extensions of our work (\Cref{sec:extensions}). We first show how the model and results can be generalized to an infinite metric, such as the full Euclidean space (e.g., with the metric $\ell_p$).
Then we discuss the implication of imposing a bound on the amount of change allowed for each individual cell; that is, each point can move by at most some given distance $\tau$. We show that the bound restriction makes the problem fundamentally harder for a general (finite) metric, since testing whether \e{any} repair exists is already NP-complete (even for a single label or two labels with a simple inclusion constraint). On the other hand, we devise a polynomial-time algorithm for finding an optimal repair for the line metric. 

\subparagraph*{Related Work.} 
We focus on related work on update repairs. For a broader view of database repair, see the excellent books \cite{DBLP:series/synthesis/2011Bertossi,DBLP:series/synthesis/2012Fan}.

Upper and lower bounds have been established for the number of changes (which corresponds to the discrete metric) under functional dependencies~\cite{DBLP:journals/tods/LivshitsKR20,DBLP:conf/icdt/KolahiL09}, which are incomparable with our coincidence constraints. Gilad, Imber and Kimelfeld~\cite{DBLP:conf/icdt/GiladIK23} studied a relational model where each cell has a distribution over possible values (and cells are probabilistically independent), and the goal is to find a most probable instantiation that satisfies integrity constraints; one can translate an optimal repair in our model to a most probable world in their model, yet the metric properties are ignored in their work and, again, their study is restricted to functional dependencies. 

Several studies investigated the computation of an optimal repair with a general distance function (that does not necessarily obey the distance/metric axioms); yet, their results are restricted to lower bounds (hardness) and heuristic algorithms without quality guarantees~\cite{DBLP:conf/icde/ChuIP13,DBLP:conf/sigmod/BohannonFFR05}. While Chu et al.~\cite{DBLP:conf/icde/ChuIP13} focused on denial constraints (which are again incomparable to our coincidence constraints), the work of Bohannon et al.~\cite{DBLP:conf/sigmod/BohannonFFR05} is closer to our work, as they considered inclusion dependencies (in addition to functional dependencies). 
Bertossi et al.~\cite{DBLP:journals/is/BertossiBFL08} studied the complexity of optimal repairs (``Database Fix Problem'') for numerical values under the Euclidean space (square of differences), and focused on denial constraints and aggregation constraints; their results for this problem are lower bounds (NP-completeness). 

Finally, there is a weaker relevance of this work to frameworks where the repairs themselves (rather than the database values) are points of the metric space~\cite{DBLP:journals/amai/ArieliDB07,DBLP:journals/fss/ArieliZ16}, and work on
consistent query answering (i.e., determining whether all repairs agree on a query answer) under keys and foreign keys~\cite{DBLP:conf/pods/HannulaW22}.  The work of Kaminsky, Pena and Naumann~\cite{DBLP:journals/pacmmod/KaminskyPN23} combines inclusion constraints and similarity between values within the problem of constraint discovery (with equality replaced by similarity); their work is relevant here in the sense that it can be combined with ours in a larger flow that improves data quality by the discovery of (violated) integrity constraints, which are then handled by a repairing phased.

\section{Formal Framework}
\label{sec:framework}

We first describe our formal framework, from the basic concepts to the problem definition. 
\subsection{Metric Spaces}
Recall that a \emph{metric space} is a pair $(M,\delta)$ where $M$ is a set of \emph{points} and $\delta:M\times M\rightarrow\reals$ is a function satisfying the following: 
    \e{positivity:} $\delta(x,y)\geq 0$ for all $x,y\in M$ and $\delta(x,y)=0$ if and only if $x=y$;
    \e{symmetry:} $\delta(x,y)=\delta(y,x)$ for all $x,y\in M$; 
    and
    \e{the triangle inequality:} $\delta(x,z)\leq \delta(x,y)+\delta(y,z)$ for all $x,y,z\in M$. 
When $M$ is finite, the metric $(M,\delta)$ can also be seen as having its distances correspond to the length of the shortest paths in an undirected graph, specifically a complete graph on $M$ where each edge $(u,v)$ has weight $\delta(u,v)$.

\begin{figure}[t]
\small
\centering
    \begin{tabular}[t]{ll}
        \multicolumn{2}{R}{\rel{Patient} (P)}\\
        \toprule
        \underline{\att{pid}}        & \att{name}   \\
        \midrule
        \val{437}     & \val{Anna}    \\
        \val{487} & \val{Bill}        \\
        \val{719}   & \val{Carl}  \\
        \val{799}    & \val{Darcy}  \\
        \bottomrule
    \end{tabular}
    \quad\quad
    \begin{tabular}[t]{ll}
        \multicolumn{2}{R}{\rel{Registration} (R)}\\
        \toprule
        \underline{\att{pid}} & \att{time}  \\
        \midrule
        \val{779} & \val{10:00}       \\
        \val{437}   & \val{13:00} \\
        \val{199}   & \val{14:00}  \\
        \bottomrule
    \end{tabular}
    \quad\quad
    \begin{tabular}[t]{ll}
        \multicolumn{2}{R}{\rel{Vaccine} (V)}\\
        \toprule
        \att{pid} & \att{nurse}   \\
        \midrule
        \val{719} & \val{018}  \\
        \val{481} & \val{017}    \\
        \val{987} & \val{078}    \\
        \bottomrule
    \end{tabular}
    \quad\quad
    \begin{tabular}[t]{ll}
        \multicolumn{2}{R}{\rel{UsedShots} (U)}\\
        \toprule
        \underline{\att{nurse}}                  & \att{\#shots}    \\
        \midrule
        \val{078} & \val{5}       \\
        \val{017} & \val{1}      \\
        \bottomrule
    \end{tabular}
    \caption{Example relations; same attribute names express intended foreign keys.}
 \label{fig:running-relations}
\end{figure}
\begin{figure}[b]
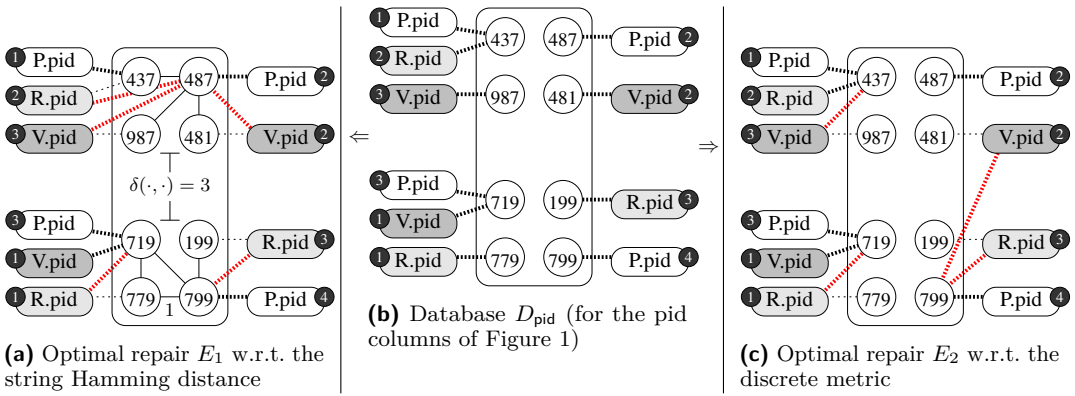

    \begin{subfigure}[t]{0.3\textwidth}
    \scalebox{0.8}{\input{running-r-sum.pspdftex}}
    \caption{\label{fig:running-sum}Optimal repair $E_1$ w.r.t.~the string Hamming distance}
    \end{subfigure}
    \;\vrule\quad
    \parbox[b]{0.3\textwidth}{
    \begin{subfigure}[t]{0.3\textwidth}
    \scalebox{0.8}{\input{running.pspdftex}}
    \caption{Database $D_{\att{pid}}$ (for the pid columns of \Cref{fig:running-relations})\label{fig:running-D}}
    \end{subfigure}
    \vskip-1em
    }
    \quad\vrule\;
    \begin{subfigure}[t]{0.3\textwidth}
    \scalebox{0.8}{\input{running-r-cnt.pspdftex}}
    \caption{\label{fig:running-cnt}Optimal repair $E_2$ w.r.t.~the discrete metric}
    
    \end{subfigure}
    \caption{\label{fig:running}Optimal repairs of a database $D_{\att{pid}}$ (middle) according to two metric spaces $(M,\delta)$ (left and right) over the person identifiers (pid).}
\end{figure}

\begin{figure}[t]
\centering
    \scalebox{0.8}{\input{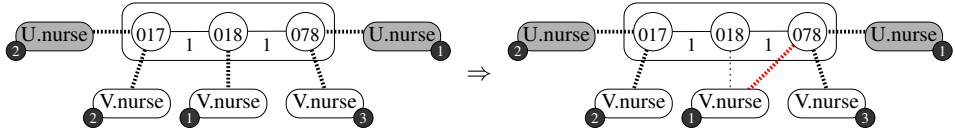}}
    \caption{\label{fig:nurse}Database $D_{\att{nurse}}$ (on the left) and an optimal repairs $E$ according to the Hamming distance and the discrete distance (on the right).}
\end{figure}

We will discuss three special cases of metric spaces $(M,\delta)$, and we will distinguish between them via a subscript of the distance function $\delta$.
\begin{itemize}
    \item A \e{line} metric $(M,\delta_{\reals})$, where $M\subseteq\mathbb{R}$ and  
    $\delta_{\reals}(x,y)=|x-y|$.
    \item A \e{discrete metric} $(M,\delta_{\neq})$, where $\delta_{\neq}(x,y)=1$ whenever $x\neq y$.
    \item A \e{tree} metric $(M,\delta_T)$, where $T$ is a weighted tree with vertex set $M$, and $\delta_T(x,y)$ is the weighted distance (i.e., the sum of weights along the unique path) in $T$ between $x$ and $y$.
\end{itemize}

\subsection{Databases}

We assume three countably infinite sets: $\atts$ is the set of all \emph{attributes}, $\cells$ is a set of all \emph{cells}, and $\vals$ is a set of \emph{values}. Each cell $c$ is labeled with an attribute $\lambda(c)\in\atts$. If $\lambda(c)=A$, then we call $c$ an \emph{$A$-cell}. By a \e{domain} we refer to a set $M\subseteq\vals$ over values.

A \emph{metric database} (or just \e{database} for short) is a mapping $D$ from a finite set of cells, denoted $\cells(D)$, to values. Hence, we view a database as a function $D:\cells(D)\rightarrow\vals$. We denote by $\vals(D)$ the set of values that occur in $D$, that is, $\vals(D)\defeq\set{D(c)\mid c\in\cells(D)}$. We say that $D$ is a database \e{over} a domain $M\subseteq\vals$ if $\vals(D)\subseteq M$. If $D$ is a database and $A\in\atts$, then we denote by $D[A]$ the restriction of $D$ to its $A$-cells; hence, $\cells(D[A])=\set{c\in\cells(D)\mid \lambda(c)=A}$, and $D[A](c)=D(c)$ for all $c\in\cells(D[A])$. If $D$ is a database and $v\in\vals$, then we denote by $D^{-1}(v)$ the set of cells $c\in\cells(D)$ with $D(c)=v$. We denote by $\atts(D)$ the set of attributes that occur in $D$, that is, $\atts(D)\defeq\set{\lambda(c)\mid c\in\cells(D)}$. 

A \emph{signature} $\signature$ is a sequence $(A_1,\dots,A_q)$ of attributes. We denote by $\atts(\signature)$ the attribute set $\set{A_1,\dots,A_q}$. An \emph{instance} of $\signature$ is a database $D$ such that $\atts(D)\subseteq\atts(\signature)$. In the remainder of this paper, we always assume the presence of a signature $\signature$, and a database $D$ that we mention is implicitly assumed to be an instance of $\signature$.

\begin{example}\label{example:running:intro}
The relational database of \Cref{fig:running-relations} 
consists of a registry of vaccinations in an improvised medical facility. (We later discuss the errors in the relations.) From the ``pid'' columns we derive the database $D_{\att{pid}}$ (in our model) of  \Cref{fig:running-D} over the signature 
$(\att{P.pid},\att{R.pid},\att{V.pid})$. Each cell $c$ is depicted by a rounded rectangle with its label $\lambda(c)$ written inside the box and its value $D_{\att{pid}}(c)$ connected to the rectangle by a dotted line. The number attached to each cell denotes the row number of the cell in \Cref{fig:running-relations}. (It  is not part of the formal model and is given here just to help connecting between the figures.) For example, the pid attribute of row~1 of \rel{Patient} gives rise to the left-top cell $c$ with $\lambda(c)=\att{P.pid}$ and 
$D_{\att{pid}}(c)=\val{437}$.
We can similarly obtain the database $D_{\att{nurse}}$ over the signature
$(\att{V.nurse},\att{U.nurse})$ from the nurse columns of the relations. This database is depicted on the left side of \Cref{fig:nurse}. We use both $D_{\att{pid}}$ and $D_{\att{nurse}}$ as running examples for this section. 
\qedexample
\end{example}

\subsection{Coincidence Constraints}
\label{example:coincidence-constraints}
Let $\signature=(A_1,\dots,A_q)$ be a signature. If $D$ is database and $v\in\vals$, then the sequence 
$p_D(v)\defeq(|D[A_1]^{-1}(v)|,\dots,|D[A_q]^{-1}(v)|)$ lists the number of $A_i$-cells with the value $v$ for $i=1,\dots,q$. We call $p_D(v)$ the \emph{coincidence profile} of $v$ (stating how many cells of each attribute coincide at $v$). 

Let $M\subseteq\vals$ be a domain. A \emph{coincidence constraint over $M$} (\emph{w.r.t.~$\signature$}) is a function $\Gamma$ that maps every value $v\in M$ to a subset $\Gamma(v)\subseteq\naturals^q$, stating the allowed coincidence profiles for every value in $M$.  A database $D$ over $M$ satisfies $\Gamma$, denoted $D\models\Gamma$, if $p_D(v)\in\Gamma(v)$ for all $v\in M$; that is, the coincidence profile of every value in $M$ is allowed by $\Gamma$. Conversely, $D$ \e{violates} $\Gamma$ if $p_D(v)\notin\Gamma(v)$ for at least one value $v\in M$, and then we denote it by $D\not\models\Gamma$.

Note that the class of coincidence constraints is closed under conjunction (intersection), disjunction (union), and negation (complement). That is, if $\Gamma_1$ and $\Gamma_2$ are coincidence constraints over $M$, then $\Gamma_1\cup\Gamma_2$ (that maps every $v\in M$ to $\Gamma_1(v)\cup\Gamma_2(v)$) is a coincidence constraint over $M$, and so are $\Gamma_1\cap\Gamma_2$ and $\naturals^q\setminus\Gamma_2$. Also note that a coincidence constraint may be infinite, but for a given database, only a finite set is relevant as the total (finite) number of relevant cells is determined by the input.


For illustration, let $\signature=(A_1,\dots,A_q)$, and let $M$ be a set of values. The \e{key constraint} $\key(A_j)$ states that no two $A_j$-cells can have the same value. Hence, $\key(A_j)$ says that every value can have at most one $A_j$-cell. The constraint $\key(A_j)$ can be expressed as the following function (which is the same on every point $v$):
$$
\Gamma_{\key(A_j)}(v)\defeq
\set{(i_1,\dots,i_q)\mid i_j\leq 1}
$$
The \e{inclusion constraint} $A_\ell\sqsubseteq A_j$ states that every value in an $A_\ell$-cell must also occur in an $A_j$-cell. Hence,
$A_\ell\sqsubseteq A_j$ states that every value that has one or more $A_j$-cells must also have one or more $A_\ell$ cells, and can be expressed by the following coincidence constraint:
$$
\Gamma_{A_\ell\sqsubseteq A_j}(v)\defeq
\set{(i_1,\dots,i_q)\mid i_\ell=0 \mbox{ or } i_j>0}
$$
Finally, the \e{foreign-key constraint} $A_\ell\fk A_j$ is the conjunction of $\key(A_j)$ and $A_\ell\sqsubseteq A_j$, hence expressed by the coincidence constraint $\Gamma_{A_\ell\fk A_j}\defeq
\Gamma_{\key(A_j)}\cap \Gamma_{A_\ell\sqsubseteq A_j}$.


A coincidence constraint $\Gamma$ may, in principle, pose a different restriction on every value of~$M$. We also consider the uniform case where $\Gamma$ is the same everywhere. Formally, we say that $\Gamma$ is \e{uniform} if $\Gamma(v)=\Gamma(u)$ for all pairs $u$ and $v$ of values in $M$. In this case, we may write just $\Gamma$ instead of $\Gamma(v)$, and treat $\Gamma$ simply as a set of coincidence profiles $(i_1,\dots,i_q)$. For example, each of $\Gamma_{\key(A_j)}$, $\Gamma_{A_\ell\sqsubseteq A_j}$ and $\Gamma_{A_\ell\fk A_j}$ is uniform since its definition does not depend on $v$, hence we can write, for instance, $\Gamma_{\key(A_j)}\defeq
\set{(i_1,\dots,i_q)\mid i_j\leq 1}$.

\begin{example}\label{example:pid-constraints}
We continue with our running example. The scenario we consider is that \Cref{fig:running-relations}
describes a relational database with a combination of clean data in the \rel{Patient} and \rel{UsedShots} relations (where administrators insert records), and noisy data in the \rel{Registration} and \rel{Vaccine} relations, where records are entered (manually or via OCR) from handwritten forms that may include mistakes or ambiguous letters.


Consider again the databases $D_{\att{pid}}$ in Figure~\ref{fig:running}(b)
derived from \Cref{fig:running-relations}. The constraints that we associate to the domain of $D_{\att{pid}}$ is 
$\att{V.pid}\fk \att{R.pid} \fk \att{P.pid}$,
stating that every vaccine is
given to a registered patient who, in turn, is in the patient registry, and no two registrations and patient records coincide in their value. For the relations of \Cref{fig:running-relations} it means that \att{pid} is a foreign key from \rel{Vaccine} to 
\rel{Registration} and from  \rel{Registration} to \rel{Patient}. In our formalism, this translates to the uniform constraint $\Gamma$, where 
$\Gamma(v)\defeq\Gamma_{\att{V.pid}\fk \att{R.pid} }(v)\cap \Gamma_{\att{R.pid}\fk \att{P.pid}}(v)$ for very value $v$. Note that the constraint is violated by $D_{\att{pid}}$ since, for example, \val{719} is a value of a V.pid-cell but not an R.pid-cell, and \val{779} is a value of an R.pid-cell but not a P.pid-cell.\qedexample
\end{example}

\begin{example}\label{example:nurse-constraints}
Still within our running example, for the domain of $D_{\att{nurse}}$ (\Cref{fig:nurse}), recall that our signature is
$(\att{V.nurse},\att{U.nurse})$. Our constraint $\Gamma$ is for this domain defined by
$$\Gamma(v)\;\defeq\;\Gamma_{\att{V.nurse}\sqsubseteq \att{U.nurse}}(v) 
\; \cap\; 
\Gamma_{\att{U.nurse}\sqsubseteq \att{V.nurse}}(v) 
\; \cap\; 
\set{(i_1,i_2)\mid i_1\leq \mbox{shots}(v)}$$
where $\mbox{shots}(v)$ is the number in the row of $v$ in the (clean) relation \rel{UsedShots} of \Cref{fig:running-relations}, that is, $\mbox{shots}(\val{078})=5$ and
$\mbox{shots}(\val{017})=1$.  Hence, $\Gamma$ states that every vaccinating nurse should occur in the registry of the used shots (i.e., $\att{V.nurse}\sqsubseteq \att{U.nurse}$) and vice versa, and the number of times a nurse is recorded as giving a vaccine 
(that is, the number of V.nurse cells per nurse identifier) is at most the number of used shots recorded for that nurse. Note that $\Gamma$ is non-uniform since it differs for the two values \val{078} and \val{017}.
It is violated since the value \val{018} has a V.pid-cell but not a U.pid-cell.
\qedexample
\end{example}

\subsection{Repairs}
Let $\signature=(A_1,\dots,A_q)$ be a signature, $(M,\delta)$ be a metric space, and $\Gamma$ be a coincidence constraint over $M$. An \e{inconsistent database} is a database $D$ over $M$ such that $D\not\models\Gamma$. A \e{repair} of $D$ is a database $E$ such that $\cells(E)=\cells(D)$ and $E\models\Gamma$. The \e{cost} of a repair $E$ is the cumulated distance that the values of $D$ undergo in the transformation to $E$. We allow to differentiate the cost between different attributes, so we assume a global weight function $w:\atts(\signature)\rightarrow\mathbb{R}_{\geq 0}$.\footnote{For example, it may be the case that we trust $A_i$-cells more than we trust $A_j$-cells, so the same movement by distance $\epsilon$ can contribute differently to the cost of $E$; this will be reflected in $w(A_i)>w(A_j)$.}
Hence, we define the cost of a repair $E$, denoted $\kappa(D,E,\delta)$, by
$$\kappa(D,E,\delta)\defeq\sum_{c\in\cells(D)}w(\lambda(c))\cdot \delta(D(c),E(c))\,.$$

\begin{example}
Continuing \Cref{example:pid-constraints},  assume that  $w(\att{P.pid})=\infty$ (or some large number), since P.pid-cells are assumed to be clean, and that $w(\att{R.pid})=1$ and $w(\att{V.pid})=1.1$ (since V-pid cells are deemed slightly more reliable than R.pid-cells).
Figures~\ref{fig:running-sum} and~\ref{fig:running-cnt} show optimal repairs of $D_{\att{pid}}$ of \Cref{fig:running-D} with respect to $(M,\delta_1)$ and 
$(M,\delta_2)$, respectively, where 
\begin{itemize}
    \item $M=\set{\val{437},\val{487},\val{987},\val{481},\val{719},\val{199},\val{779},\val{799}}$ (i.e., $\vals(D_{\att{pid}})$);
    \item $\delta_1$ is the Hamming distance between strings (e.g., $\delta_1(\val{437},\val{487})=1$ and $\delta_1(\val{437},\val{987})=2$);
    \item $\delta_2$ is the discrete distance over $M$ (e.g., $\delta_2(\val{437},\val{487})=\delta_2(\val{437},\val{987})=1$).
\end{itemize} 
Note that $\kappa(D_{\att{pid}},E_1,\delta_1)=2\cdot 1+3\cdot 1.1=5.3$, since there are two changes of R.pid and three of V.pid, each of distance $1$. (Changes are marked by red lines.) Also note that $\kappa(D_{\att{pid}},E_2,\delta_2)=4.2$ as there are two changes of R.pid and two of V.pid, but $\kappa(D_{\att{pid}},E_2,\delta_1)=7.5$ since:
\begin{align}
\att{R.pid}: \quad &1\cdot\delta_1(\val{779},\val{719})+1\cdot\delta_1(\val{199},\val{799})=1+1=2;\\
\att{V.pid}: \quad &1.1\cdot \delta_1(\val{987},\val{437})+1.1\cdot\delta_1(\val{481},\val{799})=2.2+3.3=5.5\,.
\end{align}
In particular, note that $E_2$ is worse that $E_1$ for the metric $\delta_1$.
\qedexample
\end{example}

\begin{example}
Continuing \Cref{example:nurse-constraints}, recall that $D_{\att{nurse}}$ violates $\Gamma$ because the value $\val{018}$ is associated with a V.nurse-cell but no U.pid-cells. Assume that $w(\att{U.nurse})=\infty$ since the U.nurse-cells are assumed to be clean, but $w(\att{V.nurse})=1$. Under both the discrete metric and the Hamming distance, an optimal repair is obtained by changing the cell $c$ of $\val{018}$ to $\val{078}$, as illustrated at the right of \Cref{fig:nurse}.  Note that changing the value of $c$ to \val{017} would not be legal, since it would violate the constraint $\set{(i_1,i_2)\mid i_1\leq \mbox{shots}(\val{017})}$, as $i_1$ would be two (corresponding to two V.nurse-cells with the value \val{017} while  $\mbox{shots}(\val{017})=1$. 
\qedexample
\end{example}

\subsection{Computational Problem}\label{sec:hardness}
We study the complexity of computing a low-cost repair. Formally, we assume a fixed signature $\signature$. The input consists of a finite metric space $(M,\delta)$, a finite coincidence constraint $\Gamma$ over $M$, and an inconsistent database $D$ over $M$. The goal is to compute an \e{optimal} repair, that is, a repair $E$ such that $\kappa(D,E,\delta)\leq \kappa(D,E',\delta)$ for all repairs $E'$, or declare that no repair exists. We will also study the approximate version of finding an $\alpha$-optimal repair, where $\alpha$ is a number (or a numeric function of the input), which is a repair $E$ such that $\kappa(D,E,\delta)\leq\alpha\cdot \kappa(D,E',\delta)$ for all repairs $E'$.

Note that our definition of a coincidence constraint $\Gamma$ allows the set $\Gamma(v)$ to be infinite, and examples of such constraints are shown in \Cref{example:coincidence-constraints}. However, for a given database $D$, only a finite (polynomial-size) subset of $\Gamma(v)$ is relevant---the profiles $(i_1,\dots,i_q)$ where every number $i_j$ is at most $|D[A_j]|$. Hence, the assumption that $\Gamma$ is given as part of the input is not a limitation. We will revisit this assumption when we study infinite metrics (\Cref{sec:infinite}).

It may be unclear upfront whether we can even test in polynomial time whether \e{any} repair exists (i.e., our version of the \e{existence of repair} problem~\cite{bertossi2011database}). 
It follows immediately from our later results on optimal repairs (e.g., \Cref{theorem:T}) that this problem is, indeed, solvable in polynomial time. Yet, our first result in the next section 
states hardness of approximation, even when a repair is guaranteed to exist.
\subsection{Hardness}

The next results states that for general input metrics, it is NP-hard to approximate the optimal repair beyond some fixed ratio. (See proof in \Cref{app:hardness}.)
Recall that APX-hardness means that there is some constant $\alpha>1$ such that there is no polynomial-time algorithm for computing an $\alpha$-approximation unless $\mbox{P}=\mbox{NP}$. The proof is via a PTAS reduction from the problem of finding a minimum cover by 3-sets.

\begin{restatable}{thm}{apxhard}\label{thm:apx-hard}
Let $\signature=(A_1,\dots,A_q)$ be a signature with $q\geq 2$. Minimizing the cost of a repair is APX-hard, even if the weight $w$ is uniform, $\Gamma$ is uniformly the inclusion constraint $A_1\sqsubseteq A_2$, and a repair is guaranteed to exist.
\end{restatable}


The remainder of this paper is therefore dedicated to obtaining optimal algorithms for metrics of interest (notably, the line metric and discrete metric), and providing provable approximation algorithms for general metrics.
\section{Algorithms}
In this section we provide our algorithms for our metric database repair problem. We start with an exact algorithm for the line metric and discrete metric, and (as we shall see) more generally tree metrics.
Our exact algorithm for the latter will also play a role in our approximation algorithm for general metrics in the second half of this section.

\subsection{Algorithm for Tree Metrics}
\label{sec:tree}
Many problems in tree metrics are amenable to dynamic programming approaches, allowing for polynomial-time algorithms for problems that might be intractable on general metrics. This is also the case here. In this section, we devise a polynomial-time algorithm for finding an optimal repair in the case of a tree metric. 
Formally, we will prove that:
\begin{theorem}\label{theorem:T}
    An optimal repair can be found in polynomial time (if exists), given a tree metric space $(M,\delta_T)$, a coincidence constraint $\Gamma$ over $M$, and an inconsistent database $D$.
\end{theorem}

Note that \Cref{theorem:T} implies, in particular, that we can test in polynomial time whether any repair exists; this can be done by applying the theorem to an arbitrary tree metric over the points. 
Moreover, the tractability of tree metrics implies the same for other basic metrics:

\begin{corollary}\label{cor:line-and-discrete}
An optimal repair can be found in polynomial time (if exists) in the case of a line  metric $(M,\delta_{\reals})$, and in the case of the discrete metric $(M,\delta_{\neq})$ over $M$.
\end{corollary}
\begin{proof}
Let $M=\set{v_1,\dots,v_n}$. This corollary follows from \Cref{theorem:T} by casting the two metrics as tree metrics, as illustrated in \Cref{fig:to-t}.
The case of a line metric $(M,\delta_{\reals})$ is straightforward; if (w.l.o.g.) $v_1<v_2<\dots<v_n$, then $(M,\delta_{\reals})$ is the same as the tree metric $(M,\delta_T)$ where $T$ is the path $v_1\,\mbox{---}\dots\mbox{---}\,v_n$ where the weight of each edge $\set{v_i,v_{i+1}}$ is $v_{i+1}-v_i$. In the case of the discrete metric $(M,\delta_{\neq})$ over $M$, we take as tree a star with a new vertex $v'$ as center, and leaves $v_1,\dots,v_n$. The weight of every edge is $1/2$. To make sure that (optimal) repairs do not place any cell in $v'$, we define $\Gamma(v')=(0,\dots,0)$.
\end{proof}

\begin{figure}[t]
    \input{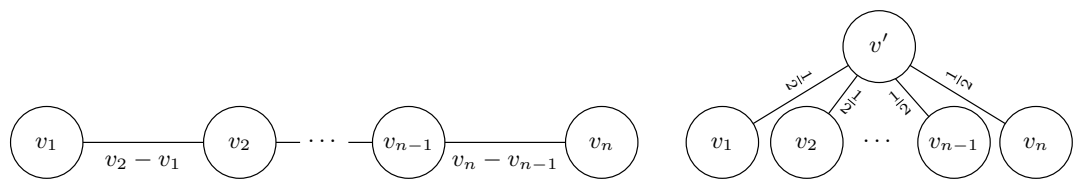}
    \caption{The line metric $(M,\delta_{\reals})$ (left) and discrete metric $(M,\delta_{\neq})$ (right) cast as tree metrics, with $M=\set{v_1,\dots,v_n}$.\label{fig:to-t}}
 \end{figure}



In the remainder of this section, we prove \Cref{theorem:T} by showing the repairing algorithm.
Throughout this section, we assume the schema $\signature=(A_1,\dots,A_q)$ and a given input $(M,\delta_T,\Gamma,D)$. We further assume that the tree $T$ itself is given (or, otherwise, we can reconstruct it from $\delta_T$). 

\subparagraph*{Translation into a binary tree.}
We first transform the tree $T$ into a binary tree where every internal vertex has precisely two children. We do so by introducing new vertices with distance zero to their parents. While this construction violates the property of non-zeroness of the metric, this does not matter for the algorithm since it is optimal even for a pseudometric (where distinct points can be of distance $0$). The end result is that every point in $M$ is a vertex of $T$, and some vertices of $T$ are not in $M$ (as we introduced them for the construction). Let $M'$ be the set of new vertices (hence, $M\subseteq M'$). We extend $\Gamma$ to $M'$ by defining $\Gamma'(v)=\Gamma(v)$ for every $v\in M$ and $\Gamma'(v)=(0,\dots,0)$ for every $v\in M'\setminus M$. This ensures that our solution places no cells in the new vertices, and so, the result is a legal repair.

In the remainder of this section, we will assume that $T$ is binary to begin with, and that $M$ and $\Gamma$ are, to begin with, the above constructed $M'$ and $\Gamma'$.

\subparagraph*{Placement.}
The algorithm deploys dynamic programming that processes $T$ bottom-up, leaf to root. For a vertex $v$ of $T$, we denote by $T_v$ the subtree of $T$ rooted at $v$, and by $D_v$ the subset of $D$ with cells having points inside $T_v$.

Let $\signature=(A_1,\dots,A_q)$.
For $j=1,\dots,q$, let $n_j$ be the number of $A_j$-cells of $D$, that is, $n_j\defeq |D[A_j]|$. Let $t_1,\dots,t_q$ be integers, where each $t_j\in[0,n_j]$. A $(t_1,\dots,t_q)$-\e{placement} in $T_v$ is a \e{consistent} database $E$ that is obtained by positioning $t_1+\dots+t_q$ cells of $D$ in $T_v$, where the number of $A_j$-cells is $t_j$. 
Some of these cells may belong to $D_v$ and others outside of $D_v$. Cells of the former kind are called \e{resident cells} and those of the latter kind are \e{visitor cells}.
The \emph{cost} of $E$ is the sum $\sum_{c\in\cells(E)\cup\cells(D_v)}w(\lambda(c))\cdot \Delta(c)$ where:
$$
\Delta(c)\defeq
\begin{cases}
    \delta_T(D(c),E(c)) & \mbox{if $c\in \cells(D_v)\cap \cells(E)$;}\\
    \delta_T(D(c),v) & \mbox{if $c\in \cells(D_v)\setminus\cells(E)$;}\\
     \delta_T(v,E(c)) & \mbox{if $c\in \cells(E)\setminus\cells(D_v)$.}
\end{cases}
$$
In words, we consider the transformation of $E$ from $D_v$; if $c$ is a cell that moves from one vertex of $T_v$ to another, then $\Delta(c)$ is the distance between the two vertices. If $c$ is a resident cell that disappears, then $\Delta(c)$ is the cost of moving $c$ to the root. If $c$ is a visitor cell that occurs in $E$, then $\Delta(c)$ is the cost of moving $c$ from the root to its vertex.

Given $t_1,\dots,t_q$, we will compute, for each vertex $v$, the minimum-cost $(t_1,\dots,t_q)$-placement (among all sets of cells) in $T_v$. This value will be stored as 
$$\opt_v(t_1,\dots,t_q)$$ 
where, as a special case, the cost of the best repair of $D$, is $\opt_r(n_1,\dots,n_q)$,
where $r$ is the root vertex of $T$.
If no $(t_1,\dots,t_q)$-placement exists, then $\opt_v(t_1,\dots,t_q)$ is $\infty$.
(As usual in dynamic programming, 
the actual repair is obtained by restoring the optimal placements that produce the least cost $\opt_r(n_1,\dots,n_q)$.)

\subparagraph*{Handling leaves.} When $v$ is a leaf, we define $\opt_v(t_1,\dots,t_q)=0$ if 
$(t_1,\dots,t_q)\in\Gamma(v)$ (hence, the coincidence profile of $v$ is legal); otherwise, $\opt_v(t_1,\dots,t_q)=\infty$. 

\subparagraph*{Handling internal vertices.}
We now consider the case where $v$ is an internal vertex. For a vertex $u$ of $T$ and $j=1,\dots,q$, we use $n_j[T_u]$ to denote the number of $A_j$-cells in the tree $T_u$ (as positioned in $D$), that is, the sum of $|(D^{-1}(u'))[A_j]|$ over all vertices $u'$ in the subtree $T_u$.

Our goal is to find an optimal $(t_1,\dots,t_q)$-placement for the entry $\opt_v(t_1,\dots,t_q)$. We will find the minimal cost for every coincidence profile 
$(i_1,\dots,i_q)\in\Gamma(v)$ of the vertex $v$, and then we will take the least-cost entry across all coincidence profiles. In the remainder of this part, we fix $(i_1,\dots,i_q)$. 

Recall that $v$ is an internal vertex, and so, has precisely two children. Let us denote them by $v_1$ and $v_2$.  To obtain our optimal $(t_1,\dots,t_q)$-placement, we can \emph{pull} from $T_{v_\ell}$ (where $\ell\in\set{1,2}$) a set of $A_j$-cells of size $p$ for $p\in\set{0,\dots,n_j[T_{v_\ell}]}$, or \emph{push} to $T_{v_\ell}$ a set $A_j$-cells of size $p$ for $p\in\set{0,\dots,t_j-n_j[T_{v_\ell}]}$. (Note that there is no gain in pulling $A_j$-cells and pushing $A_j$-cells at the same time since the placement can use the pulled cell instead of the pushed cell with no additional cost, or even a lower cost.)
We say uniformly that we pull from $T_{v_\ell}$ a set of $A_j$-cells of size $p_{\ell,j}$ for $p_{\ell,j}\in\set{-(t_j-n_j[T_{v_\ell}]),\dots, n_j[T_{v_\ell}]}$, where the meaning of pulling $-p$ cells, for $p\geq 0$, is pushing $p$ cells. 
Once we pull $A_j$-cells from (or push $A_j$-cells to) $T_{v_\ell}$, we place the cells optimally in $T_{v_\ell}$. To find the cost of that, we use a previously computed 
$\opt_{v_\ell}(t_1^\ell,\dots,t_q^\ell)$ where $t_j^\ell=n_j[T_{v_\ell}]-p_{\ell,j}$ (i.e., the number of $A_j$-cells that remain in $T_{v_\ell}$).  Hence, the total cost is given by:
\begin{equation}\label{eq:cost_enrich}
\sum_{\ell=1,2} \,
\Big(\sum_{j=1}^q 
\big(w(A_j)\cdot p_{\ell,j}\cdot \delta_T(v,v_\ell)\big)+
\opt_{v_\ell}(t_1^\ell,\dots,t_q^\ell)\Big)
\end{equation}
We will then iterate over all \emph{legal} combinations of
$p_{\ell,j}$ and take the minimal cost according to \Cref{eq:cost_enrich}. The combination is legal if, for all $j=1,\dots,q$, the number of cells that we position in $T_v$ is indeed $t_j$. This number consists of the number of $A_j$-cells that remain in each $T_{v_\ell}$, namely $n_k[T_{v_\ell}]-p_{\ell,j}$, plus the number $i_j$ of $A_j$-cells that remain in $v$; hence:
    $$i_j+p_{1,j}+p_{2,j} = t_j$$


This concludes the description of the dynamic program. Clearly, the execution of the program terminates in polynomial time. (Recall that the schema $\signature=(A_1,\dots,A_q)$ is fixed and, in particular, $q$ is treated as a constant.)
The dynamic program is correct in the sense that it constructs an optimal repair if any repair exists. This is proved by showing that the program is indeed solving the generalized optimization problem:
\begin{lemma}
Let $v$ be a vertex of $T$.  When we process $v$ with $t_1,\dots,t_q$, we compute the least cost of a $(t_1,\dots,t_q)$-placement in $T_v$, or  $\infty$ if none exists.
\end{lemma}
The proof is straightforward from the description of the algorithm in this section.


\subsection{General (Finite) Metrics}
\label{sec:tree-to-general}
Leveraging our algorithm for tree metrics and classic results for probabilistic tree embeddings, in this section we devise a logarithmic-ratio approximation for general metrics. 



Next, we prove that we can obtain a logarithmic approximation for a general (finite) metric. Formally, we will prove that:

\begin{theorem}\label{thm:general-log}
    There is a polynomial-time randomized algorithm that, given a metric $(M,\delta)$, a coincidence constraint $\Gamma$, an inconsistent database $D$, and an error probability $\epsilon>0$, finds an $O(\log|M|)$-optimal repair with probability at least $1-\epsilon$, if any repair exists.
\end{theorem}
In the remainder of this section, we prove \Cref{thm:general-log}. 
We fix the input $(M,\delta,\Gamma,D)$ 
for the rest of this section.
The proof is based on
 the following classic \emph{tree embedding} lemma~\cite{DBLP:journals/jcss/FakcharoenpholRT04} that allows for the translation of exact algorithms for tree metrics to $O(\log |M|)$-approximation algorithms for arbitrary metrics.

\begin{lemma}[\!\!\cite{DBLP:journals/jcss/FakcharoenpholRT04}]\label{lemma:tree-embedding}
    Given a metric $(M,\delta)$, there exists a polytime samplable distribution $\mathcal{P}$ over weighted trees $T$ defining tree metrics $(M,\delta_T)$ satisfying for every $u,v\in M$:
    \begin{enumerate}
        \item $\delta(u,v)\leq \delta_T(u,v)$ for every $T$ in the support of $\mathcal{P}$.
        \item $\mathbb{E}_{T\sim \mathcal{P}}[\delta_T(u,v)] = O(\log |M|)\cdot \delta(u,v)$.
    \end{enumerate}
\end{lemma}
Blelloch, Gu, and Sun~\cite{blelloch2017efficient} show a near-linear-time counterpart of \Cref{lemma:tree-embedding}. We note that randomization and expectation are crucial for this theorem; for example, if we embed an $n$-point cycle with unit-distance edges in any tree, then at least one pair of vertices will suffer an $\Omega(n)$ distortion~\cite[Theorem 7]{DBLP:conf/focs/Bartal96}. 
%
%

By combining \Cref{theorem:T} and \Cref{lemma:tree-embedding}, we will
show how we obtain an $O(\log |M|)$ approximation, with high probability, in polynomial time. 
We do so by repeatedly finding an optimal repair for multiple random trees $T$, and taking the best outcome. More precisely, to obtain an 
$O(\log |M|)$-approximation with probability at least $1-\epsilon$, we take the best outcome out of the $O(\log\frac1\epsilon)$ repetitions of the following procedure:
\begin{enumerate}
    \item Select a random tree $T$ according to \Cref{lemma:tree-embedding}.
    \item Find an optimal repair $E_T$ for $(M,\delta_T)$, $\Gamma$ and $D$.
\end{enumerate}

\def\Eopt{E_{\mathrm{opt}}}

We show that the process gives an $O(\log|M|)$-approximation with probability at least $1-\epsilon$. Let us denote by $\Eopt$ an optimal repair for the original metric $(M,\delta)$. 
We will use the following two lemmas.

\begin{restatable}{lemma}{expectedlog}\label{lemma:expected-log}
$\mathbb{E}_{T\sim \mathcal{P}}[\kappa(D,E_T,\delta)] \leq O(\log |M|)\cdot \kappa(D,\Eopt,\delta)$.
\end{restatable}
\begin{proof}[Proof (Sketch)]
We prove the lemma (in \Cref{app:proof-of-expectation}) by analyzing $\mathbb{E}_{T\sim \mathcal{P}}[\kappa(D,E_T,\delta)]$, applying the linearity of expectation and \Cref{lemma:tree-embedding}. 
\end{proof}

\begin{lemma}\label{lemma:probable-log}
$\mathrm{Pr}_{T\sim \mathcal{P}}
\Big[\kappa(D,E_T,\delta) \leq C\cdot\log |M|\cdot \kappa(D,\Eopt,\delta)\Big]\geq \frac12$ for some constant $C>0$. 
In words, $E_T$ is $O(\log |M|)$-optimal with probability at least $1/2$.
\end{lemma}
\begin{proof}
This follows immediately from the inequality of  \Cref{lemma:expected-log} and Markov's inequality, namely
$\mathrm{Pr}[X\geq a]\leq \frac{\mathbb{E}[X]}{a}$, where $\kappa(D,E_T,\delta)$ plays the role of $X$ and  $2\cdot O(\log |M|)\cdot \kappa(D,\Eopt,\delta)$ plays the role of $a$.
\end{proof}

We can now complete the proof of \Cref{thm:general-log}, since the randomized procedure can be seen as a Bernoulli trial that succeeds (i.e., produces a good approximation) with a probability of at least $1/2$; hence, we see at least one success with probability $1-\epsilon$ after
$\log_2\frac1\epsilon = O(\log \frac1\epsilon)$ independent trials.

\section{Extensions}
In this section, we study two extensions of our study: the case of an infinite metric, and bound restriction on the movement of each individual cell.

\label{sec:extensions}
\subsection{Infinite Metrics}\label{sec:infinite}
Up to now, we have considered databases over a finite metric $(M,\delta)$ that is given explicitly as part of the input. In particular, a repair could use only values from the given point set $M$. There are, however, natural situations where the metric is a known \e{infinite} metric that can provide additional points for repairs. We wish to be able to repair an inconsistent database by using arbitrary values from the infinite metric.
An example is the $\ell_p$-metric $(M,\delta)$ where $M$ is the Euclidian space $\mathbb{R}^k$ and $\delta$ is defined by the norm $\Vert \cdot \Vert_p$ over $M$, that is, $\delta(v,u)=\Vert v-u \Vert_p$. 

To model the computational problem in the case of infinite metrics, we consider the case where the metric $(M,\delta)$ is fixed and infinite. Moreover, the coincidence constraint $\Gamma$ is fixed, and we restrict the discussion to a uniform  $\Gamma$ (that maps every point in $M$ to the same, possibly infinite, set of coincidence profiles). We will further assume that $\Gamma$ contains the profile $(0,\dots,0)$ since, otherwise, it is impossible to satisfy $\Gamma$ using any database, as our databases are finite. Computationally, we only require polynomial-time computation of distances $\delta(u,v)$, given $u$ and $v$, and membership testing in $\Gamma$, given $(i_1,\dots,i_q)$. 

The following theorem states that the ability to use points outside of $D$ is not useful if we are satisfied with a 2-approximation and the coincidence constraint is closed under addition. Note that a uniform coincidence constraint $\Gamma$ over $M$ is \e{closed under addition} if for every pair
$(i_1,\dots,i_q)$ and $(i'_1,\dots,i'_q)$ of  profiles in $\Gamma$, the profile $(i_1+i'_1,,\dots,i_q+i'_1,)$ is also in $\Gamma$. For illustration, referring to \Cref{example:coincidence-constraints}, the inclusion constraint $\Gamma_{A_j\sqsubseteq A_\ell}$ is closed under addition, but the key constraint $\Gamma_{\key(A_j)}$ and the foreign-key constraint $\Gamma_{A_j\fk A_\ell}$ are \e{not} closed under addition.

\begin{restatable}{prop}{noninventive}\label{prop:non-inventive}
Let $(M,\delta)$ be an infinite metric space, $\Gamma$ a uniform coincidence constraint, and $D$ an inconsistent database. If $\Gamma$ is closed under addition, then there exists a 2-optimal repair $E$ such that $\vals(E)\subseteq\vals(D)$.
\end{restatable}
\begin{proof}[Proof (Sketch)] To establish $E$, we take every point in $\vals(E)\subseteq\vals(D)$ and move all cells in that point to the nearest cell. The distance to each moved cell can then grow at most twice. See \Cref{app:infinite}.
\end{proof}

Note that it is necessary to make an assumption on the coincidence constraint in \Cref{prop:non-inventive}. If we remove the assumption, the statement is false simply because there may be no repair at all over the domain of $D$. For example, if $\Gamma$ is the key constraint $\key(A_j)$, then it may be necessary to introduce new metric points if $D$ has fewer points than $A_j$-cells.

\Cref{prop:non-inventive} implies that we can reduce the case of an infinite metric to the case of a finite one (assuming that the coincidence constraint is closed under addition). In particular, we can apply \Cref{thm:general-log} to conclude a logarithmic approximation in the infinite case.

\begin{theorem}\label{thm:general-log-infinite}
Let $(M,\delta)$ be an infinite metric space, and let $\Gamma$ be a uniform coincidence constraint that is closed under addition. There is a polynomial-time randomized algorithm that, given an inconsistent database $D$ and an error probability $\eta>0$, finds an $O(\log|M|)$-optimal repair with probability at least $1-\eta$. 
\end{theorem}

The assumption on $\Gamma$ is crucial for \Cref{thm:general-log-infinite}. For example, suppose that $(M,\delta)$ is the $\ell_p$-metric and $\Gamma$ is a key constraint. In that case, we can construct a repair with an arbitrarily small cost, by generating enough points around those of $D$. Hence, any $\alpha$-optimal solution must be an optimal solution, since our approximation ratio is multiplicative. In particular, we cannot get any approximation guarantee by using the tree embedding of \Cref{lemma:tree-embedding}.

\subsection{Bound Restriction}
\label{sec:bound}
\def\trueT{\mathsf{T}}
\def\falseF{\mathsf{F}}

Note that a low cost of a repair $E$ does not necessarily imply that an individual cell is moved to a point that is close to its origin. This is due to our choice to measure the cost of a repair as the sum of cell movements. It is clearly of interest to consider the variant of the problem where we limit the movement of individual cells by a threshold.

In this section, we consider the extension of our repairing problem where a bound is posed on the maximal movement of a cell. Formally, consider a metric space $(M,\delta)$, a coincidence constraint $\Gamma$, and an inconsistent database $D$. For a threshold $\tau>0$, a \e{$(\delta\leq\tau)$-repair} is a repair $E$ such that $\delta(D(c),E(c))\leq w(\lambda(c))\cdot\tau$ for every $c\in\cells(D)$. Our goal is now to find a $(\delta\leq\tau)$-repair $E$ with a minimal $\kappa(D,E,\delta)$.

However, the bound restriction is nontrivial to deal with. Our algorithms inherently fail to deal with this restriction. The dynamic-programming algorithm of \Cref{theorem:T} is based on the fact that we can remember the total movement of cells from/to the root, but not any information about individual cells. Moreover, the approximation using the random tree embedding of \Cref{lemma:tree-embedding} cannot lead to the support of a bound restriction, since every random tree may (unavoidably) have a high distortion, meaning that two individual points $u$ and $v$ can be way farther in $\delta_T$ than in $\delta$. 

In fact, even the existence of a repair (regardless of its cost) becomes an intractable problem in the presence of a bound constraint.

\begin{restatable}{prop}{boundedhard}\label{prop:bounded-hard}
It is NP-complete to determine whether any $(\delta\leq\tau)$-repair exists, given a (finite) metric $(M,\delta)$, a coincidence constraint $\Gamma$, an inconsistent database $D$, and threshold~$\tau$. The problem remains NP-hard in each of the following cases:
\begin{enumerate}
    \item The signature $\signature$ consists of a single attribute and $\Gamma$ is uniform.
     \item The signature is $\signature=(A_1,A_2)$ and $\Gamma$ is the inclusion constraint $\Gamma_{A_1\sqsubseteq A_2}$.
\end{enumerate}
\end{restatable}

\begin{proof}[Proof (Sketch)]
For the first case we devise a reduction from exact cover by 3-sets, and for the second we use CNF satisfiability. See \Cref{app:bounded}.
\end{proof}
\subsubsection{Algorithm for the Line Domain}
In contrast to the hardness shown in \Cref{prop:bounded-hard}, we can efficiently force a bound restriction in the case of a line metric.
\begin{theorem}\label{thm:line-bounded-ptime}
An optimal $(\delta\leq \tau)$-repair can be found in polynomial time, given $\Gamma$, $D$, $\tau$ and a line metric $(M,\delta)$. The same holds true if $(M,\delta)=(\mathbb{R},\delta_{\mathbb{R}})$ is fixed and $\Gamma$ is fixed, uniform, and closed under addition.
\end{theorem}
In the remainder of this section, we prove \Cref{thm:line-bounded-ptime} by presenting an algorithm for computing an optimal $(\delta\leq\tau)$-repair. Throughout the section, we assume that $M\subseteq\mathbb{R}$ is a set of numbers and $\delta$ is $\delta_{\mathbb{R}}$. We consider separately the cases where $M$ is finite and given as part of the input, and where $M$ is $\mathbb{R}$ itself.

\subparagraph*{Given (finite) metric.} 
We first introduce some notation. Let $D$ be a database. 

A \e{subset} of $D$ is a database is $D'$ such that $\cells(D')\subseteq\cells(D)$ and $D'(c)=D(c)$ for all $c\in\cells(D')$. We write $D'\subseteq D$ to state that $D'$ is a subset of $D$, and we denote by $D\setminus D'$ the complement of $D'$, that is, the subset $D''$ of $D$ with $\cells(D'')=\cells(D)\setminus\cells(D')$. If $D'\subseteq D$ and $E$ is a repair of $D$, then we denote by $E_{|D'}$ the subset $E'$ of $E$ with $\cells(E')=\cells(D')$. 

A \e{prefix} of $D$ is a database that comprises a prefix of $D$'s cells for each attribute; that is, it is a database  $D'\subseteq D$ such that if $c\in\cells(D'[A_j])$ for some attribute $A_j$, then for all $c'\in\cells(D[A_j])$ with $D(c')<D(c)$, it holds that $c'\in\cells(D'[A_j])$.
We write $D'\subseteqp D$ to denote that $D'$ is a prefix of $D$. 
We say that $D'$ is a \e{strict} prefix 
of $D$ if $D'\subseteqp D$ and $\cells(D')\subsetneq \cells(D)$. If $D'$ is a prefix of $D$, then the complement $D\setminus D'$ is a \e{suffix} of $D$.

We say that $D$ is \e{contracted} if $D$ consists of a single point, that is, $D(c)=D(c')$ for all $c$ and $c'$ in $\cells(D)$. For a set $C\subseteq\cells(D)$, we say that $C$ \e{satisfies} $\Gamma$ (denoted $C\models\Gamma$) if $D\models\Gamma$ for a contracted database $D$ with $\cells(D)=C$. (Note that either all such contracted $D$ or none of them satisfy $\Gamma$, since the common cell value has no impact on satisfying $\Gamma$.)

The following lemma implies that, without loss of generality, we can assume that an optimal repair consists of two parts, as illustrated in \Cref{fig:bounded-1-anatomy}: one is an optimal repair of a strict prefix, and the other is a contraction of the suffix. 

\begin{figure}
  \centering
  \hskip1em
  \input{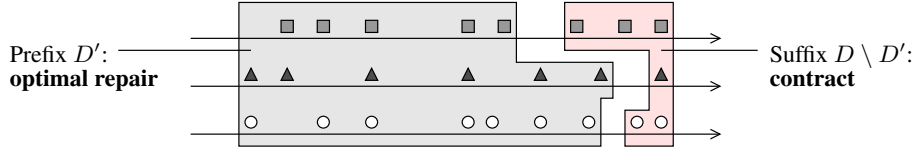}
  \caption{A database $D$ over the line metric. Each shape corresponds to a cell of one of three labels (circle, triangle, square). The structure of an optimal repair by \Cref{lemma:line-suffix-contracted} comprises of two: an optimal repair for a strict prefix, and a contracted suffix.}
  \label{fig:bounded-1-anatomy}
  \end{figure}

\begin{restatable}{lemma}{linesuffixcontracted}
    \label{lemma:line-suffix-contracted}
Let $D$ be an inconsistent database. If any $(\delta\leq\tau)$-repair exists, then there is an optimal $(\delta\leq\tau)$-repair $E$ and a strict prefix $D'$ of $D$ with the following properties.
\begin{enumerate}
    \item $E_{|D'}$ is an optimal $(\delta\leq\tau)$-repair of $D'$.
    \item $E_{|D\setminus D'}$ is contracted.
\end{enumerate}
\end{restatable}
\begin{proof}[Proof (Sketch)]
The proof (in \Cref{app:bounded}) is based on showing that two $A$-cells cannot cross each other in their movement from $D$ to $E$. 
\end{proof}

Note that $D$ has a polynomial number of prefixes. Also
note that $\subseteqp$ is a partial order over the prefixes of $D$. Hence, due to \Cref{lemma:line-suffix-contracted}, we can apply dynamic programming to compute an optimal $(\delta\leq\tau)$-repair of every prefix $D'$ of $D$, where we traverse the prefixes in a topological order according to  $\subseteqp$, starting with the empty set.

We compute as follows the cost $\kappa(D',E',\delta)$ of an optimal $(\delta\leq\tau)$-repair $E'$ of $D'$. (By recording the decisions throughout the procedure, we can derive the repair itself.)
Let $v_1,\dots,v_n$ be the values of $M$, with $v_i<v_j$ for all $i<j$.  
For a prefix $D'$ of $D$, we denote by $V^{D'}_r$ the cost of a $(\delta\leq\tau)$-repair $E'$ of $D'$ such that $E'(c)\in\{v_1,\dots,v_r\}$ for all $c\in\cells(D)$, and $\kappa(D',E',\delta)$ is minimal among all $(\delta\leq\tau)$-repairs of $D'$ satisfying this property.
Our final goal is to compute the value $V^D_n$.
 We show how to compute $V^{D'}_r$ using dynamic programming. 
     \begin{enumerate}
            \item If $\cells(D')=\emptyset$, then $V^{D'}_r=0$. 
            \item  If $\cells(D')\neq\emptyset$ and $r<1$, then $V^{D'}_r=\infty$.
            \item Otherwise, let $\mathcal{P}$ be the set of all prefixes $D''$ of $D'$ such that $\cells(D'\setminus D'')\models\Gamma$ and 
            $\delta(D'(c),v_r)\le w(\lambda(c))\cdot\tau$ for all $c\in\cells(D'\setminus D'')$. Then: $$V^{D'}_r=\min_{D''\in\mathcal{P}}\Big( V^{D''}_{r-1}+\sum_{c\in\cells(D'\setminus D'') }w(\lambda(c))\cdot\delta(D'(c),v_r)\Big).\,$$
       \end{enumerate}
   The first case refers to the situation where $D'$ is empty; hence, the cost is zero. The second case refers to the situation where we no longer have available positions for the remaining cells, and the cost is infinite (hence, there is no $(\delta\leq\tau)$-repair). In the third case, we go over all possible prefixes $D''$ of $D'$ with $\cells(D'\setminus D'')\models\Gamma$. In this case, the repair $E'$ is such that $E'(c)\in\{v_1,\dots,v_{r-1}\}$ for all $c\in \cells(D'')$, while $E'(c')=v_r$ for every cell $c'\in\cells(D'\setminus D'')$. In this case, $V^{D''}_{r-1}$ is the minimal cost for $D''$, and to that we add the cost of changing the  cells of $D'\setminus D''$ to $v_r$. Since $E'_{|D'\setminus D''}$ is contracted, by checking that $\cells(D'\setminus D'')\models\Gamma$, we ensure that we do not violate consistency by placing all cells of $D'\setminus D''$ together. Then, $V^{D''}_{r-1}$ is responsible for checking that $E'_{|D''}\models\Gamma$. This guarantees that we obtain a repair. Furthermore, we only consider the prefixes $D''$ such that $\delta(D'(c),v_r)\le w(\lambda(c))\cdot\tau$ for all $c\in\cells(D'\setminus D'')$, and $V^{D''}_{r-1}$ is responsible for checking that $\delta(D'(c),E'(c))\le w(\lambda(c))\cdot\tau$ for all $c\in\cells(D'')$, which guarantees that we obtain a $(\delta\leq\tau)$-repair.
    
      The correctness of the algorithm is a direct consequence of Lemma~\ref{lemma:line-suffix-contracted}. Regarding the computational complexity, we can compute each value $V^{D'}_r$ in polynomial time, since, as aforementioned, we can compute the set of all prefixes of $D'$ in polynomial time. Moreover, it is rather straightforward that we can compute the summation in polynomial time, and also verify the conditions that $\cells(D'\setminus D'')\models\Gamma$ and $\delta(D'(c),v_r)\le w(\lambda(c))\cdot\tau$ for all $c\in\cells(D'\setminus D'')$. Finally,
    there are polynomially many values $V^{D'}_r$ that we need to compute. Therefore, we conclude that the algorithm runs in polynomial time in $|D|$.

\subparagraph*{The full line.} 
When $M=\mathbb{R}$, we can use the following lemma, which gives a reduction to the finite case that we discussed in the previous part.

\begin{restatable}{lemma}{lineinfinitetofinite}\label{lemma:line-infinite-to-finite}
Suppose that $M=\mathbb{R}$ and $\Gamma$ is closed under addition. Let $D$ be an inconsistent database. If any $(\delta\leq\tau)$-repair exists, then there is an optimal $(\delta\leq\tau)$-repair $E$ such that
$\vals(E)\subseteq\left(\vals(D)\cup\set{(v\pm (w(A)\cdot\tau)\mid v\in\vals(D),A\in\atts(D)}\right)$.
\end{restatable}
\begin{proof}[Proof (Sketch)]
We show (in \Cref{app:bounded}) that if the repair $E$ uses a point  outside the stated domain, then the entire set of cells in that point can be moved to a point in the domain without increasing the cost; this move is  possible since $\Gamma$ is closed under addition.
\end{proof}

\section{Conclusions}
We studied the problem of finding an optimal repair of an inconsistent database (i.e., a set of labeled cells) with respect to a coincidence constraint. We established that incorporating the metric space underlying the domain of values can lead to algorithms with efficiency and quality guarantees. In summary: the problem is APX-hard for general metrics but logarithmically approximable in polynomial time, and moreover the problem is solvable optimally in polynomial time for a tree metric (hence, for the common line and discrete metrics). 
We also discussed the case of an infinite metric and the addition of bound constraints. The addition of the bound constraints makes it NP-hard to test whether any legal repair exists, but an optimal repair can be found in polynomial time for the line metric. 

Many directions are left for future work. First, for general metrics our lower bound is APX-hardness (i.e., some constant ratio) while the upper bound is logarithmic; how can this gap be closed? 
Next, as mentioned in \Cref{sec:extensions}, we have left open the case of a tree metric with bound restrictions. Note, however, that even if this case can be solved optimally in polynomial time, it is not at all clear that this tractability has implications on other metrics (e.g., via embedding) as it has in the absence of bound constraints. Still, it is an important challenge to find natural metrics, beyond the line metric, where nontrivial upper bounds can be established. 

Another major direction for future work is to extend our work to constraints besides coincidence constraints, including others commonly studied for data quality management: functional dependencies (and their conditional enrichment~\cite{DBLP:conf/icde/BohannonFGJK07}), denial constraints, non-unary inclusion constraints, and so on. Another direction is the extension to coincidence constraints with a non-fixed (given) set of labels (i.e., the ``combined complexity'' variant of this problem), which requires a formalism for compactly expressing coincidence constraints.  

Finally, it is important to investigate the practical aspects of our work. How do the algorithms of this paper perform on common datasets? How practical is the dynamic program for the tree metric? How can we optimize it? What is the actual approximation ratio that takes place in a general metric? How does it change from one metric to another? These questions call for a careful implementation and experimental investigation, as next steps.






\bibliographystyle{plainurl}
\bibliography{refs}

\appendix
\section{Deferred Proof of \Cref{sec:hardness} (APX-Hardness for General Metrics)}\label{app:hardness}

\apxhard*

\begin{proof}
We will show a PTAS reduction from the problem of finding a minimum cover by 3-sets, which we denote here shortly as 3SC. The input to this problem consists of a set $X$ of $m$ elements, and a collection $S$ of $n$ subsets of $X$, each of size $3$. The goal is to find a minimal subset $S'$ of $S$ that covers $X$, that is, $\cup S'=X$. This problem is known to be APX-hard~\cite{DBLP:journals/ipl/Kann91}.

Given the input $(X,S)$ for 3SC, we construct an instance of our problem as follows. The metric $(M,\delta)$ is defined as $M\defeq X\cup S\cup\set{r}$ where $r$ is a new point, and $\delta$ is defined by the undirected graph that has the following edges, each of unit distance, as illustrated in the left part of \Cref{fig:3set-apx}: there is an edge between $r$ and every $s\in S$, and between $x\in X$ and $s\in S$ whenever $x\in s$. As stated in the theorem, the coincidence constraint is $\Gamma_{A_1\sqsubseteq A_2}$. The database contains an $A_1$-cell $c_x$ for every $x\in X$, and an $A_2$-cell $c_s$ for every 
$s\in S$. 
The positions are $D(c_x)=x$ for all $x\in X$, and $D(d_s)=r$ for $s\in S$. Finally, $w(A_j)=1$ for $j=1,\dots,q$.
This concludes the construction.

First, we show how we translate a repair $E$ into a cover $S'$. The simple case is where every $A_1$-cell $c_x$ is in a point $s$ such that $x\in s$. If $E$ has this property, then we say that $E$ is a \e{cover repair}. In this case, we know that each metric point $s$ that includes an $A_1$-cell also includes at least one $A_2$-cell (due to the constraint $A_1\sqsubseteq A_2$), and so, we take as the cover $S'$ the set of all sets $s\in S$ that include one or more $A_2$-cells. 

Note, however, that a repair is not necessarily a cover repair. For illustration, \Cref{fig:3set-apx} shows examples of a cover repair and a non-cover repair.

\begin{figure}[h]
  \centering
  \input{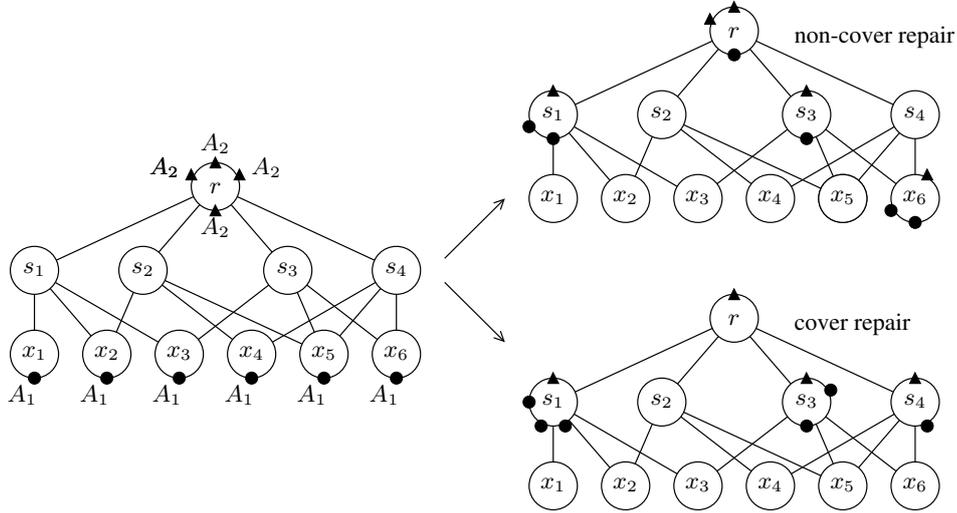}
  \caption{\label{fig:3set-apx}
  Reduction in the proof of  \Cref{thm:apx-hard}. A big circle represents a point (graph vertex) in the metric $M$. A small (filled) circle represents an $A_1$-cell of the database $D$, and a triangle represents an $A_2$-cell. The construction is for the 3SC instance with $X=\set{x_1,\dots,x_6}$ and $S$ consisting of the sets $s_1=\set{x_1,x_2,x_3}$, $s_2=\set{x_2,x_4,x_5}$,  
  $s_3=\set{x_3,x_5,x_6}$ and $s_4=\set{x_4,x_5,x_6}$. Left: constructed instance. Top right: non-cover repair. Bottom right: cover repair (assuming each circle comes from a neighboring vertex).}
\end{figure}

Given a repair $E$, we construct a cover repair $E'$ as follows. In the first step, we handle every point $x\in X$ that includes cells. We reposition each cell $c_{x'}$ where $x'\neq x$, in an arbitrary point $s$ such that $x\in s$. We are then left with, possibly, the $A_1$-cell $c_x$ and a nonempty set of $A_2$-cells. So, in the second step, we move all of these to an arbitrary point $s$ such that $x\in s$. Note that, at this time, there may be points $s\in S$ that violate $A_1\sqsubseteq A_2$ since they may include $A_1$-cells that we moved there, but no $A_2$-cells. So, in the third step, we reposition all of the $A_2$-cells: we place one $A_2$-cell in every nonempty $s$ point and leave all of the rest in the root $r$.

We need to prove that we have not increased the cost of the repair, that is, $\kappa(D,E',\delta)\leq\kappa(D,E,\delta)$. To see that, observe the following. In the first step, we gain at least one unit for every cell $c_{x'}$ that we are moving, since the distance between $x'$ and $x$ is at least two whereas we position $c_{x'}$ in a point of distance one from $x'$. In the second step, we pay one unit for repositioning $x$, but gain at least one for the repositioning of the $A_2$-cells (since their distance from $r$ is reduced from two to one). In the third step, we pay one unit for every cell $s$ that is empty in $E$ and nonempty in $E'$; however, for each such cell $s$ we gain one unit from every cell that has been moved to $s$ in the first step. All in all, we gain at least as much as we pay.

From the construction so far we conclude that the 3SC instance has a solution of cost $p$ if and only if there is a repair of cost $m+p$. To complete the proof, we need to show that, given $\beta>1$, there is $\alpha>1$ so that an $\alpha$-optimal repair entails a $\beta$-optimal cover. Let $p^\star$ be the cost of a minimal set cover. If we have an $\alpha$-optimal repair $E$ of cost $m+p$, then we have
$m+p\leq\alpha\cdot (m+p^\star)$, and so, $p\leq (\alpha-1)\cdot m+\alpha\cdot p^\star$. Note, however, that $m$ is at most of size $3p^\star$, and so
$p\leq (\alpha-1)\cdot 3p^\star+\alpha\cdot p^\star
= (4\alpha-3)\cdot p^\star
$.
Hence, to obtain a $\beta$ approximation, we need to take $\alpha=(\beta+3)/4$.
\end{proof}

\section{Deferred Proof of \Cref{sec:tree-to-general} (Repair via Tree Embedding)}
\label{app:proof-of-expectation}

\expectedlog*
\begin{proof}
Recall from \Cref{lemma:tree-embedding} that for every random tree $T$, each distance in $\delta_T$ is mapped to a smaller or equal distance in $\delta$; hence, $\kappa(D,E_T,\delta)\leq\kappa(D,E_T,\delta_T)$. Also observe that $\Eopt$ is a repair, and so,
$\kappa(D,E_T,\delta)\leq\kappa(D,\Eopt,\delta_T)$ due to the optimality of $E_T$. Hence, we conclude the following:
\begin{align*}
\mathbb{E}_{T\sim \mathcal{P}}[\kappa(D,E_T,\delta)] & 
\leq \mathbb{E}_{T\sim \mathcal{P}}[\kappa(D,E_T,\delta_T)] \\
& \leq \mathbb{E}_{T\sim \mathcal{P}}[\kappa(D,\Eopt,\delta_T)]\\
& = \mathbb{E}_{T\sim \mathcal{P}}\Big[\sum_{c\in\cells(D)}w(\lambda(c))\cdot\delta_T(D(c),\Eopt(c))\Big]\\
&=
\sum_{c\in\cells(D)}w(\lambda(c))\cdot\mathbb{E}_{T\sim \mathcal{P}}\Big[\delta_T(D(c),\Eopt(c))\Big]\\
& \leq \sum_{c\in\cells(D)}w(\lambda(c))\cdot
O(\log |M|)\cdot \delta(D(c),\Eopt(c))\\&=
O(\log |M|)\cdot \sum_{c\in\cells(D)}w(\lambda(c))\cdot\delta(D(c),\Eopt(c))\\
&
= O(\log |M|)\cdot \costsum(D,\Eopt,\delta),
\end{align*}
where the equalities follow by definitions and linearity of expectation, while the last inequality is again due to~\Cref{lemma:tree-embedding}. 
\end{proof}

\section{Deferred Proof of \Cref{sec:infinite} (Infinite Metric)}
\label{app:infinite}

\noninventive*
\begin{proof}
Let $E_0$ be a repair. We will construct a repair $E$ such that
$\vals(E)\subseteq\vals(D)$ and $\kappa(D,E,\delta)\leq2\cdot\kappa(D,E_0,\delta)$. To construct $E$, we update every cell with a value $v$ outside of $\vals(D)$.
For that, we apply the following change to $E_0$, for every $v\in\vals(E_0)\setminus\vals(D)$, until no such $v$ exists. 
Let $c_m\in E_0^{-1}(v)$ be a cell such that $\delta(D(c_m),v)$ is minimal among the cells $c\in E_0^{-1}(v)$. Hence, $c_m$ is the cell that has the least cost of update among those of $E_0^{-1}(v)$. We replace in $E_0$ the value of every cell in $E_0^{-1}(v)$ with the value $D(c_m)$. 

Observe that $E$ is consistent, since $E$ is obtained by placing, in each metric point, either no cells or the disjoint union of cell collections from multiple points of $E_0$. The former is allowed by the assumption that $(0,\dots,0)\in\Gamma$, and the latter is allowed by the assumption that $\Gamma$ is closed under addition. 

Note that for every $c\in E_0^{-1}(v)$ it holds that
$$\delta(D(c),E(c))
\leq \delta(D(c),E_0(c))+ \delta(E_0(c),E(c))
\leq 2\cdot \delta(D(c),E_0(c))\,.
$$
The first inequality is due to the triangle inequality. The second inequality is due to the construction above where  $\delta(E_0(c),E(c))$ is the same as $\delta(v,D(c_m))$, and $\delta(v,D(c_m))=\delta(D(c_m),v)$ is at most $\delta(D(c),v)=\delta(D(c),E_0(c))$ due to symmetry.  We conclude that 
$\delta(D(c),E(c))$ is at most $2\cdot \delta(D(c),E_0(c))$ and then the proposition immediately follows.
\end{proof}

\section{Deferred Proof of \Cref{sec:bound} (Bound Restriction)}
\label{app:bounded}

\boundedhard*

\begin{proof}
Membership in NP is straightforward. We will show reductions for the two cases, where we represent each $(M,\delta)$ as a graph. See \Cref{fig:bounded-reductions} for an illustration of the reductions.

\begin{figure}
  \def\xt{x^\trueT}
   \def\xf{x^\falseF}
  \input{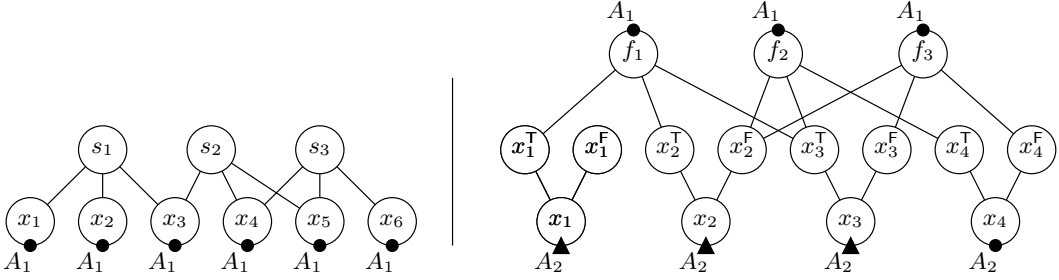}
  \caption{Reductions in the proof of  \Cref{prop:bounded-hard}. A big circle represents a point (graph vertex) in the metric $M$, a small (filled) circle represents an $A_1$-cell, and a triangle represents an $A_2$-cell. The left graph (Case~1) is constructed for the X3C instance with $X=\set{x_1,\dots,x_6}$, $s_1=\set{x_1,x_2,x_3}$, $s_2=\set{x_3,x_4,x_5}$,   and $s_3=\set{x_4,x_5,x_6}$. The right graph (Case~2) is constructed for the CNF formula
  $\varphi=f_1\land f_2\land f_3$ where $f_1=x_1\lor x_2\lor x_3$, $f_2=\neg x_2\lor x_3\lor x_4$, and $f_3=\neg x_2\lor \neg x_3\lor\neg x_4$. 
  \label{fig:bounded-reductions}}
 \end{figure}

\subparagraph*{Case 1.} We devise a reduction from the problem of \e{exact cover by 3-sets} (X3C). The input to this problem consists of a set $X$ of $3m$ elements, and a collection $S$ of subsets of $X$ of size $3$. The goal is to determine whether there is a cover of $X$ that consists precisely $m$ sets from $S$ (hence, the sets are necessarily pairwise disjoint).  

Given the input $(X,S)$ of X3C, we construct an instance of our problem as follows. The metric $(M,\delta)$ is defined as $M\defeq X\cup S$ and $\delta$ is defined by the undirected graph that has an edge (of a unit weight) between $x$ and $s$ whenever $x\in s$. The coincidence constraint $\Gamma$ states that there are either three or zero cells in each point; that is, $\Gamma=\set{(0),(3)}$. 
The database $D$ consists of a single cell for each $x\in X$. Finally, $w(A_1)=1$ and $\tau=1$ (i.e., every cell can pass one edge).

To complete the proof, we need to show that a $(\delta\leq\tau)$-repair exists if and only if an exact cover exists. If $E$ is a repair, then it can have cells only in $S$-points; otherwise, an $X$-point should have three elements, and then two of these should have gone through at least two edges, in contrast to the bound constraint $\tau=1$. Hence, each nonempty $S$-point $s$ contains precisely three cells, and these cells represent the elements of $s$. We conclude that the structure of $E$ is an exact cover. On the other hand, we can similarly construct a $(\delta\leq\tau)$-repair from every exact cover of $(X,S)$.

\subparagraph*{Case 2.} We devise a reduction from SAT. We are given a CNF formula $\varphi=f_1\land\dots\land f_m$ over a set
$\set{x_1,\dots,x_n}$ of variables, where each $f_i$ is a disjunction of atomic formulas of the form $x_i$ or $\neg x_i$. The goal is to determine whether $\varphi$ has a satisfying assignment. 

Given $\varphi$, we construct an instance of our problem as follows. The set $M$ is the union of three sets:
$$M\defeq \set{f_1,\dots,f_m}\cup\set{x_1,\dots,x_n}\cup \set{x_i^b \mid i\in\set{1,\dots,n}\land b\in\set{\trueT,\falseF}}$$
The distance $\delta$ is defined using the graph that connects each $x_i$ to $x_i^{\trueT}$ and
$x_i^{\falseF}$, each $f_j$ to $x_i^{\trueT}$ whenever $x_i$ occurs positively in $f_j$, and each $f_j$ to $x_i^{\falseF}$ whenever $x_i$ occurs negatively in $f_j$. All weights are unit weights. The coincidence constraint is $\Gamma_{A_1\sqsubseteq A_2}$, as required, and $w(A_1)=w(A_2)=\tau=1$. The database $D$ is defined as follows. We have a unique $A_1$-cell in each $f_j$ position, and a unique $A_2$-cell in each $x_i$ cell.

To complete the proof, we need to show that a $(\delta\leq\tau)$-repair exists if and only if a satisfying assignment exists. Let $E$ be a $(\delta\leq\tau)$-repair. Note that each $A_1$-cell corresponds to a conjunct $f_j$, and it must be positioned in a position $x_i^b$ that contains an $A_2$-cell, which must be $x_i$. Hence, $x_i$ is moved to one of the $x_i^b$, and specifically one that satisfies $f_j$. We conclude that the $A_2$-cells in the $x_i^b$ positions encode a satisfying assignment. The reverse direction, where we construct a $(\delta\leq\tau)$-repair from a satisfying assignment, is constructed similarly to the first direction.
\end{proof}

\linesuffixcontracted*
\begin{proof}
Let $E$ be an optimal $(\delta\leq\tau)$-repair of $D$. Let $v$ be the largest value of $M$ such that $E(c)=v$ for some $c\in\cells(D)$. Let $D'$ be the subset of $D$ that contains all the cells $c\in\cells(D)$ with $E(c)\neq v$. We first show that $D'$ is a strict prefix of $D$. Assume, by way of contrdiction, that $D'$ is not a strict prefix of $D$. We have that $\cells(D')\subsetneq \cells(D)$ by definition; hence, the only possible case is that there exists an attribute $A_j$ and two cells $c,c'\in\cells(D[A_j])$ such that $D(c)<D(c')$ while $c'\in \cells(D')$ and $c\not\in\cells(D')$. This implies that $E(c)=v$ (as we only exclude from $D'$ cells that satisfy this property), while $E(c')<v$ (as $v$ is the largeast value of $M$ occuring in $E$). We conclude that $D(c)<D(c')$ while $E(c)>E(c')$. We can now construct a repair $E'$ of $D$ by swapping between these two values; that is, defining $E'(c)=E(c')$ and $E'(c')=E(c)$ while keeping the values of the other cells intact. It is rather straightforward that $E\models\Gamma$ implies $E'\models\Gamma$. Moreover, $E'$ is a $(\delta\leq\tau)$-repair and $\kappa(D,E,\delta)>\kappa(D,E',\delta)$, since $\delta(D(c),E(c))>\delta(D(c),E'(c))$ and $\delta(D(c'),E(c'))>\delta(D(c'),E'(c'))$, and $E$ and $E'$ agree on the values of the rest of the cells. This is a contradiction to the fact that $E$ is an optimal $(\delta\leq\tau)$-repair of $D$. We conclude that $D'$ is a strict prefix of $D$, and it is now easy to see that $E_{|D'}$ is an optimal $(\delta\leq\tau)$-repair of $D'$; otherwise, we can construct a repair $E'$ of $D$ such that $E'_{|D'}$ coincides with an optimal $(\delta\leq\tau)$-repair of $D'$ and $E'(c)=v$ for all $c\not\in\cells(D')$, with $\kappa(D,E,\delta)>\kappa(D,E',\delta)$. Finally, $E_{|D\setminus D'}$ is contracted by construction, and that concludes our proof.
\end{proof}

\lineinfinitetofinite*
\begin{proof}
Let $E$ be a $(\delta\leq\tau)$-repair of $D$. Let $c\in\cells(D)$ be a cell such that $E(c)\not\in V$ for $V=\vals(D)\cup\set{(v\pm (w(A)\cdot\tau)\mid v\in\vals(D),A\in\atts(D)}$. If no such cell exists, then $E$ satisfies the desired property and that concludes our proof. Otherwise, let $c_1,\dots,c_n\in\cells(D)$ be all the cells such that $E(c_\ell)=E(c)$ (including the cell $c$ itself). We assume, without loss of generality, that for every $\ell\in\{1,\dots,n-1\}$, it holds that $D(c_\ell)\le D(c_{\ell+1})$.
We show how we can transform $E$ into a $(\delta\leq\tau)$-repair $E'$, such that $E'(c_\ell)\in V$ for every $\ell\in\{1,\dots,n\}$, without increasing the cost. By repeatedly applying this process to each group of cells with a value outside $V$, we can obtain a $(\delta\leq\tau)$-repair $E''$ such that $E''(c)\in V$ for all $c\in\cells(D)$ and $\kappa(D,E'',\delta)\le\kappa(D,E,\delta)$.

It cannot be the case that $E(c)$ is located to the right of $D(c_n)$; otherwise, by defining $E'(c_\ell)=D(c_n)$ for every $\ell\in\{1,\dots,n\}$ (and $E'(c)=E(c)$ for any other cell), we clearly do not introduce any violations of the constraint (since $\Gamma$ is closed under addition), and for every $\ell\in\{1,\dots,n\}$ we have that: \[\delta(D(c_\ell),E'(c_\ell))= \delta(D(c_\ell),E(c_\ell))-\delta(D(c_n),E(c_n)).\]
Hence, $\kappa(D,E',\delta)\le\kappa(D,E,\delta)$, which is a contradiction to the fact that $E$ is an optimal $(\delta\leq\tau)$-repair. Similarly, it cannot be the case that $E(c)$ is located to the left of $D(c_1)$, or we can define $E'(c_\ell)=D(c_1)$ for all $\ell\in\{1,\dots,n\}$, and obtain a repair $E'$ with $\kappa(D,E',\delta)\le\kappa(D,E,\delta)$. 

Therefore, $E(c)$ is located between $D(c_k)$ and $D(c_{k+1})$ for some $k\in\{1,\dots,n-1\}$. We consider two possible cases.
\begin{enumerate}
    \item If $\sum_{\ell=1}^k w(\lambda(c_\ell))\ge \sum_{\ell=k+1}^n w(\lambda(c_\ell))$, let $c_i\in\set{c_{k+1},\dots,c_n}$ be the cell with the highest $D(c_i)-w(\lambda(c_i))\cdot\tau$ value among all these cells (if there are several cells with the same value, we choose one arbitrarily). In this case, we define $E'(c_\ell)=D(c_i)-w(\lambda(c_i))\cdot\tau$ for all $\ell\in\set{1,\dots,n}$ and $E'(c)=E(c)$ for all other cells $c$. Note that it must be the case that $D(c_i)-w(\lambda(c_i))\cdot\tau\le E(c)$; otherwise, we have that $\delta(D(c_i),E(c))>w(\lambda(c_i))\cdot\tau$, which contradicts the fact that $E$ is a $(\delta\leq\tau)$-repair. Thus,
    we have that $\delta(D(c_\ell),E'(c_\ell))\le \delta(D(c_\ell),E(c_\ell))$ for all $\ell\in\set{1,\dots,k}$ and since $E$ is a $(\delta\leq\tau)$-repair, we have that $\delta(D(c_\ell),E'(c_\ell))\le w(\lambda(c_\ell))\cdot\tau$ for all $\ell\in\set{1,\dots,k}$. As for the cells in $\set{c_{k+1},\dots,c_n}$, since $D(c_i)-w(\lambda(c_i))\cdot\tau\ge  D(c_\ell)-w(\lambda(c_\ell))\cdot\tau$ for all $\ell\in\set{k+1,\dots,n}$ (due to our choice of $c_i$), it holds that: \[\delta(D(c_\ell),D(c_i)-w(\lambda(c_i))\cdot\tau)\le \delta(D(c_\ell),D(c_\ell)-w(\lambda(c_\ell))\cdot\tau)=w(\lambda(c_\ell))\cdot\tau\] for all $\ell\in\set{k+1,\dots,n}$; hence, $E'$ is a $(\delta\leq\tau)$-repair of $D$. Let $d=\delta(E(c), D(c_i)-w(\lambda(c_i))\cdot\tau)$. It holds that: 
    \[\kappa(D,E',\delta)=\kappa(D,E,\delta)-\sum_{\ell=1}^k w(\lambda(c_\ell))\cdot d+\sum_{\ell=k+1}^n w(\lambda(c_\ell))\cdot d\]
    and since $\sum_{\ell=1}^k w(\lambda(c_\ell))\ge \sum_{\ell=k+1}^n w(\lambda(c_\ell))$, we conclude that $\kappa(D,E',\delta)\le\kappa(D,E,\delta)$.
    \item If $\sum_{\ell=1}^k w(\lambda(c_\ell))< \sum_{\ell=k+1}^n w(\lambda(c_\ell))$, this case is symmetric to case (1).
\end{enumerate}
Note that since $\Gamma$ is closed under addition, we never introduce new violations of the constraints by moving a set of cells from one point to another. In both cases we obtain a $(\delta\leq\tau)$-repair $E'$ of $D$ with $\kappa(D,E',\delta)\le\kappa(D,E,\delta)$, and this concludes our proof. 
\end{proof}

\end{document}